\documentclass[letterpaper, 10 pt, conference]{ieeeconf} 
\IEEEoverridecommandlockouts
\usepackage{amsmath,amssymb}
\usepackage{graphicx}
\usepackage{bm}
\usepackage{svg}
\usepackage{textcomp}
\usepackage{xcolor}
\usepackage{cite}
\usepackage{url}
\usepackage[hidelinks]{hyperref}
\usepackage{algorithm}
\usepackage[noend]{algorithmic}
\usepackage{xprintlen}
\usepackage{cases}
\usepackage{tabularx,ragged2e}
\usepackage{booktabs}
\usepackage{nccmath}
\usepackage{caption}
\usepackage{float}
\usepackage[labelformat=simple]{subcaption}
\usepackage{mathtools}
\usepackage{verbatim}
\usepackage{import}

\def\BibTeX{{\rm B\kern-.05em{\sc i\kern-.025em b}\kern-.08em
    T\kern-.1667em\lower.7ex\hbox{E}\kern-.125emX}}
\usepackage{array}
\newcolumntype{P}[1]{>{\centering\arraybackslash}p{#1}}

\newtheorem{problem}{\textbf{Problem}}

\newtheorem{definition}{\bf{Definition}}
\newtheorem{lemma}{\bf{Lemma}}
\newtheorem{theorem}{\bf{Theorem}}
\newtheorem{remark}{\bf{Remark}}
\newtheorem{assumption}{\bf{Assumption}}

\def\myunderbar#1{\underline{\sbox\tw@{$#1$}\dp\tw@\z@\box\tw@}}

\newcommand{\vect}[1]{\boldsymbol{#1}}
\newcommand{\norm}[1]{ \lVert #1 \rVert}

\newcommand{\Nhop}[2]{\mathcal{N}_{#1}^{{#2}\text{-hop}}} 
\newcommand{\Neigh}[2]{\mathcal{N}_{#1}^{{#2}}} 
\newcommand{\estimate}[3]{\hat{{#1}}^{{#2}}_{{#3}}}
\newcommand{\state}[3]{{{#1}}^{{#2}}_{{#3}}}
\newcommand{\error}[3]{\Tilde{{#1}}^{{#2}}_{{#3}}}
\newcommand{\statediff}[5]{{#1}^{#2}_{#3}-{#4}^{#5}_{#3}}

\newcommand{\kc}[2]{{#1}_{#2}^{\text{kc}}}


\begin{document}

\title{Robust Estimation and Control for Heterogeneous Multi-agent Systems Based on Decentralized $k$-hop Prescribed Performance Observers}

\author{Tommaso Zaccherini, Siyuan Liu and Dimos V. Dimarogonas
	\thanks{This work was supported in part by the Wallenberg AI, Autonomous Systems and
Software Program (WASP) funded by the Knut and Alice Wallenberg (KAW) Foundation, Horizon Europe EIC project SymAware (101070802), 
the ERC LEAFHOUND Project, and the Swedish Research Council (VR).
}
\thanks{Tommaso Zaccherini and Dimos V. Dimarogonas are with the Division
of Decision and Control Systems, KTH Royal Institute of Technology, Stockholm, Sweden.
	E-mail: {\tt\small \{tommasoz, dimos\}@kth.se}. Siyuan Liu is with the Department of Electrical Engineering, Control Systems Group, Eindhoven University of Technology, the Netherlands. Email: {\tt\small s.liu5@tue.nl}
}
}

\maketitle
\begin{abstract}
We propose decentralized $k$-hop Prescribed Performance State and Input Observers for heterogeneous multi-agent systems subject to bounded external disturbances. In the proposed input/state observer, each agent estimates the state and input of agents located two or more hops away using only local information exchanged with $1$-hop neighbors, while guaranteeing that transient estimation errors satisfy predefined performance bounds. Conditions are established under which the input observer can be omitted, allowing the state observer convergence to be independent of the input estimates.
Theoretical analysis demonstrates that if a closed-loop controller with full state knowledge achieves the control objective and the estimation-based closed-loop system is set–Input to State Stable (set-ISS) with respect to the goal set, then the estimated states can be used to achieve the system objective with an arbitrarily small \mbox{worst-case} error governed by the accuracy of the states estimates. Simulation results are provided to validate the proposed approach.
\end{abstract}


\section{Introduction}

Heterogeneous \mbox{multi-agent} systems (MAS) consist of multiple autonomous agents with diverse dynamics, sensing, and computational capabilities that cooperate to achieve common objectives \cite{eb184279-05f5-3acc-a300-750c6f4a17e8}. Unlike homogeneous MAS, where identical agents limit adaptability, heterogeneous configurations integrate complementary resources—such as aerial-ground collaboration, distributed sensing and computation—to accomplish complex missions with enhanced efficiency, robustness, and fault tolerance. This diversity, however, increases coordination and estimation challenges, especially under limited communication or sensing. Rather than assuming perfect global state sharing, enabling each agent to estimate the state of other agents beyond its immediate neighbors can significantly improve cooperative performance and resilience.

Research on distributed estimation and observer-based control for MAS has produced a variety of approaches \cite{ZHAO201322, XU2013334, LI2011510, ACIKMESE20141037, 7440802, 6308694, 11187218,11018605, 8264361}. \mbox{Observer-based} control schemes \cite{ZHAO201322, XU2013334} generally achieve consensus or tracking for specific system classes but are tailored to particular control objectives and lack theoretical guarantees when integrated with other controllers. Although distributed observers and \mbox{consensus-based} filters \cite{ACIKMESE20141037, 6308694,8264361} enable local estimation through neighbor communication, they often assume homogeneous and \mbox{disturbance-free} settings, lack predefined estimation performance guarantees, and typically require each agent to reconstruct the full network state, thereby limiting scalability in large networks.

To overcome these limitations, our previous work \cite{11018605} introduced a \mbox{$k$-hop} Distributed Prescribed Performance Observer (\mbox{$k$-hop} DPPO) for homogeneous, \mbox{disturbance-free} MAS, enabling each agent to estimate the state of agents that are two---or more---hops away using only $1$-hop communication. While this approach guarantees predefined estimation performance, it relies on \mbox{network-dependent} gain tuning and prior knowledge of input estimation error bounds, which are often difficult to determine in \mbox{large-scale} or complex networks and usually require centralized information. Moreover, it is formulated for homogeneous disturbance-free systems which limits its applicability to realistic heterogeneous scenarios.

Motivated by these challenges and inspired by Prescribed Performance Control (PPC) \cite{4639441}, we propose decentralized $k$-hop Prescribed Performance State and Input Observers (\mbox{$k$-hop} PPSO and \mbox{$k$-hop} PPIO) for heterogeneous MAS subject to bounded disturbances. The proposed observers enable each agent to estimate the state and input of agents located two to $k$ hops away, while ensuring that the estimation errors satisfy predefined performance specifications set at the design stage.
 Unlike conventional distributed observers \cite{11187218,11018605,6308694,8264361}, the proposed framework is fully decentralized, relying solely on local (\mbox{$1$-hop}) communication without requiring any  global network information, input bounds, or assumptions of homogeneous agent dynamics. This purely local interaction ensures scalability and facilitates deployment in large heterogeneous networks.
 Moreover, the prescribed performance formulation inherently guarantees robustness against bounded disturbances and model uncertainties, ensuring desired transient and \mbox{steady-state} behavior of the estimation errors.
Beyond the observer design, we identify conditions under which the state observer can be simplified by removing the \mbox{$k$-hop} PPIO. Finally, we show that feedback controllers ensuring \mbox{set-ISS} stability of the \mbox{closed-loop} system can preserve their control objectives, with an arbitrarily small \mbox{worst-case} error, even when nonlocal state information is replaced by locally estimated counterparts.

The remainder of the paper is organized as follows.
Section~\ref{Preliminaries} introduces the notation, preliminaries, and problem formulation.
Section~\ref{Disagreement dynamics} defines the disagreement vectors among the agents' estimates and derives their dynamics.
Section~\ref{state observer} presents the proposed \mbox{$k$-hop} Prescribed Performance State Observer and the conditions under which it can be simplified.
Section~\ref{input observer} introduces the \mbox{$k$-hop} Prescribed Performance Input Observer.
Section~\ref{Closed-loop} describes the feedback control structure and establishes the conditions under which the \mbox{$k$-hop} \mbox{estimation-based} feedback controller guarantees convergence to the team objective.
Section~\ref{Simulations} demonstrates the effectiveness of the proposed approach through simulation results, and Section~\ref{Conclusion and Future work} concludes the paper with final remarks and directions for future work.

\par

\section{Preliminaries and Problem Setting}\label{Preliminaries}
\textbf{Notation:} Denote by $\mathbb{R}$, $\mathbb{R}{\ge 0}$, and $\mathbb{R}{> 0}$ the sets of real, nonnegative, and positive real numbers, respectively. $\mathbb{R}^n$ represents the $n$-dimensional Euclidean space, and $\mathbb{R}^{n \times m}$ denotes the set of real matrices with $n$ rows and $m$ columns. Denote by $I_n$ the identity matrix of size $n$ and by $1_n$ the vector of ones of size $n$. Let $|S|$, $S^c$ and $\partial S$ be the cardinality, the complement and the boundary of a set $S$, and denote with $\bigtimes^N _{i=1} S_i$ the Cartesian product of $N$ sets $\{S_1,\dots, S_N\}$. Furthermore, let $\max_{i\in\{1,\dots,n\}} \{s_i\}$ and $\min_{i\in\{1,\dots,n\}} \{s_i\}$ denote the maximum and minimum element in a set $S  = \{s_1, \dots, s_n\}$, respectively.
Given a symmetric matrix $B \in \mathbb{R}^{n\times n}$, we represent with $\lambda_{\text{min}}(B)$ and $\lambda_{\text{max}}(B)$ respectively the minimum and maximum eigenvalues of $B$, we use $B \succ 0$ to denote a positive definite matrix $B$, and $\norm{B}$ to denote the spectral norm of $B$. Given $x \in \mathbb{R}^n$, $\norm{x}= \sqrt{x^\top x}$. 
Let $\text{diag}(a_1, \dots, a_n)$ be the diagonal matrix with diagonal elements  $a_1, . . . , a_n$ and let $\otimes$ be the Kronecker product. 
We use $f \in \mathcal{C}_1$ to denote that a function $f$ is continuous differentiable in its domain.
We define functions $\mathcal{K}$ and $\mathcal{KL}$ as follows: $\mathcal{K} \!=\! \{\gamma : \mathbb R_{\ge 0}\rightarrow\mathbb R_{\ge 0}   :  \gamma \text{ is continuous, strictly increasing and } \gamma(0)=0\}$;  $\mathcal{KL} \!=\! \{\beta : \mathbb R_{\ge 0} \!\times \mathbb R_{\ge 0} \rightarrow\mathbb R_{\ge 0}  :$ for each fixed $s$, the map  $\beta(r,s) \in \mathcal{K}$  with  respect to  $r$  and, for each fixed  nonzero $r$,  the map $\beta(r,s)$ is decreasing with respect to  $s$  and $\lim_{ s \rightarrow \infty } \beta(r,s) = 0\}$.

\subsection{Multi-agent systems}

Consider a heterogeneous MAS consisting of a set of $N$ interacting agents $\mathcal{V} = \{1,\dots, N\}$. Denote with $\vect{x}$ and $\vect{u}$ the global state and input of the system and suppose each agent $i \in \mathcal{V}$ evolves as:
\begin{equation}
    \label{eq: agent's dynamic}
    \dot x_i(t) =  f_i(x_i(t)) + g_i(u_i(t)) + w_i(\vect{x},t),
\end{equation}
where $x_i \in \mathbb{R}^{n_i}$ and $u_i \in \mathbb{R}^{m_i}$ are the state and input of agent $i$, respectively, $f_i:\mathbb{R}^{n_i} \rightarrow \mathbb{R}^{n_i}$ is the flow drift, $g_i:\mathbb{R}^{m_i} \rightarrow \mathbb{R}^{n_i}$ is an input function, and $w_i: \bigtimes^N _{i=1}\mathbb{R}^{n_i} \times \mathbb{R}_{\geq 0} \rightarrow \mathbb{R}^{n_i}$ represents external disturbances acting on $i$.

Let $n = \sum_{i=1}^{N} n_i$ and $m = \sum_{i=1}^{N} m_i$ be the dimensions of the global state and input. Then, $\vect{x}= \left[x_1^\top, \dots, x_N^\top \right]^\top \in \mathbb{R}^n$ and $ \vect{u}=\left[u_1^\top, \dots, u_N^\top\right]^\top \in \mathbb{R}^m$.
\begin{assumption}
    \label{Assumption on existence of a solution}
    (i) $f_i:\mathbb{R}^{n_i} \rightarrow \mathbb{R}^{n_i}$ is locally Lipschitz; (ii) $g_i:\mathbb{R}^{m_i} \rightarrow \mathbb{R}^{n_i}$ is measurable and essentially locally bounded; (iii) $w_i: \bigtimes^N _{i=1}\mathbb{R}^{n_i} \times \mathbb{R}_{\geq 0} \rightarrow \mathbb{R}^{n_i}$ is continuous and uniformly bounded in $\bigtimes^N _{i=1}\mathbb{R}^{n_i} \times \mathbb{R}_{\geq 0}$.
\end{assumption}
\begin{assumption}
    \label{Assumption on g}
    One of the following holds: (i) $g_i$ is bounded; (ii) $\dot{g}_i$ is bounded.
\end{assumption}
\begin{remark}
    To ensure convergence of the input observer in Section~\ref{input observer} without requiring Assumption~\ref{Assumption on g}-(i) to hold, $g_i(u_i)$ in \eqref{eq: agent's dynamic} is assumed independent of $x_i$.
    Yet, as stated in Remark \ref{Remark on input observer avoided}, under Assumption \ref{Assumption on g}-(i),  $g_i$ can be extended to $g_i(x_i, u_i)$ while preserving the state observer convergence.
\end{remark}

\subsection{Communication graph} 
The interactions among agents are represented by an undirected graph \(\mathcal{G} = (\mathcal{V}, \mathcal{E})\), where \(\mathcal{V}\) is the set of agents, and \(\mathcal{E} \subseteq \mathcal{V} \times \mathcal{V}\) is the set of communication links. An edge $(i,j) \in \mathcal{E}$ indicates that agents $i$ and $j$ can exchange information.
A path between two agents $i, j \in \mathcal{V} $ is defined as a sequence of non-repeating edges connecting $i$ to $j$. Then, a \mbox{$k$-hop} path is a path of length $k$ connecting $i$ and $j$. 

For each agent $i \in \mathcal{V}$, let $\Nhop{i}{k}$ denote the set of \mbox{$k$-hop} neighbors of agent $i$, i.e., of all nodes $N^i_j \in \mathcal{V}$ from which there exists a $p$-hop path to $i$ with $2 \leq p \leq k$. The $\eta_i = |\Nhop{i}{k}|$ elements of this set are denoted as $\Nhop{i}{k}= \{ N^i_1, \dots, N^i_{\eta_i}\}$, where $N^i_j \in \mathcal{V}$, with $j \in \{1, \dots, \eta_i\}$, indicates the global index of the $j$-th \mbox{$k$-hop} neighbor of $i$.
For simplicity, we use \(\mathcal{N}_i\) to indicate the set of direct ($1$-hop) neighbors of agent \(i\), excluding $i$ in presence of self-loops.

Suppose the following assumption hold:
\begin{assumption}
\label{Assumption on neighbors}
    $\mathcal{G}$ is a time invariant undirected graph and each $i\in \mathcal{V}$ knows its neighborhood $\Neigh{i}{}$ and $\Nhop{i}{k}$.
\end{assumption}
\begin{assumption}
    \label{Assumption on 2-hop communication}
    Each agent $i\in \mathcal{V}$ has access and can relay, at each time instant, the state and input of its \mbox{$1$-hop} neighbors $j \in \Neigh{i}{}$ to $\Neigh{i}{}$.
\end{assumption}

Assumption~\ref{Assumption on neighbors} is not restrictive, as distributed neighborhood discovery algorithms have been extensively studied in the sensor network literature \cite{994183}.
Furthermore, Assumption~\ref{Assumption on 2-hop communication} is satisfied in scenarios where each agent can measure the states of its $1$-hop neighbors using onboard sensors and share this information with its neighbors.
\begin{remark}
    \label{Remark on validity of Assumption 4}
    If Assumption \ref{Assumption on 2-hop communication} does not hold, as will be explained later in Remark \ref{Remark on local disagreement vector computation and positive definiteness of Mik}, the proposed approach can still be applied by including the direct \mbox{$1$-hop} neighbors in the definition of \mbox{$k$-hop} neighbors. Note that, by definition of $\Neigh{i}{}$, $i$ is excluded also from its new $k$-hop neighborhood. 
\end{remark}

\subsection{State and input estimates}
\label{Section: state and input estimate}
Let $\vect{x}^i$ and $\state{\vect{g}}{i}{}$ denote the stack vectors containing the state and input function of the \mbox{$k$-hop} neighbors of agent $i$, i.e., of $N_j^i \in \Nhop{i}{k}$:
\begin{equation}
    \label{stack vector of real values estimated by agent i}
    \state{\vect{x}}{i}{} = \left[\state{x}{\top}{N_1^i}, \dots ,\state{x}{\top}{N^i_{\eta_i}} \right]^\top, \ \state{\vect{g}}{i}{} = \left[\state{g}{\top}{N_1^i}, \dots ,\state{g}{\top}{N^i_{\eta_i}} \right]^\top
\end{equation}
and let $\estimate{\vect{x}}{i}{} = \left[\estimate{x}{i \ \top}{N_1^i}, \dots ,\estimate{x}{i \ \top}{N^i_{\eta_i}} \right]^\top$ and $\estimate{\vect{g}}{i}{} = \left[\estimate{g}{i \ \top}{N_1^i}, \dots ,\estimate{g}{i \ \top}{N^i_{\eta_i}} \right]^\top$
be their estimates carried out by the agent~$i$, i.e., $\estimate{x}{i}{N_j^i}$ and $\estimate{g}{i}{N_j^i}$, for $N_j^i \in \Nhop{i}{k}$, are the estimates of the state $\state{x}{}{N_j^i}$ and input function $\state{g}{}{N_j^i}$ of agent $N_j^i$ done by $i$. Moreover, denote with  $\tilde{\vect{x}}^i$ and $\tilde{\vect{g}}^i$ the corresponding estimation errors:
\begin{equation}
    \label{Definition: error on input and state estimation perfromed by i}
    \error{\vect{x}}{i}{} =\left[\error{x}{i \ \top}{N_1^i}, \dots ,\error{x}{i \ \top}{N^i_{\eta_i}}\right]^\top, \error{\vect{g}}{i}{} =\left[\error{g}{i \ \top}{N_1^i}, \dots ,\error{g}{i \ \top}{N^i_{\eta_i}}\right]^\top,
\end{equation}
where $\error{x}{i}{N_j^i} = \estimate{x}{i}{N_j^i} -\state{x}{}{N_j^i}$ and  $\error{g}{i}{N_j^i} = \estimate{g}{i}{N_j^i} -\state{g}{}{N_j^i}$ for all $N_j^i \in \Nhop{i}{k}$.

Let $\vect{x}_i$ and $\vect{g}_{i}$ be the vectors defined as $\vect{x}_i = 1_{\eta_i} \otimes x_i$ and $\vect{g}_{i} = 1_{\eta_i} \otimes g_i(u_i)$,
and let:
\begin{equation}
    \label{eq: estimation of state and input of agent i}
    \hspace{-0.2cm}\estimate{\vect{x}}{}{i} = \left[\estimate{x}{N_1^i \top}{i}, \dots , \estimate{x}{N_{\eta_i}^i \top}{i} \right]^\top\hspace{-0.1cm},
    \estimate{\vect{g}}{}{i} = \left[\estimate{g}{N_1^i \top}{i}, \dots , \estimate{g}{N_{\eta_i}^i\top}
    {i} \right]^\top\hspace{-0.1cm},
\end{equation}
be the stacked vectors containing the estimates of $x_i$ and $g_i$ computed by the \mbox{$k$-hop} neighbors of agent $i$, i.e.,
$\estimate{x}{N_j^i}{i}$ and $\estimate{g}{N_j^i}{i}$, for $j \in \{1, \dots, \eta_i\}$, are the estimates of $x_i$ and $g_i$ performed by agent $N_j^i \in \Nhop{i}{k}$.

As in \eqref{Definition: error on input and state estimation perfromed by i}, indicate with $\error{\vect{x}}{}{i} = \estimate{\vect{x}}{}{i} -\vect{x}_i$ and $\error{\vect{g}}{}{i} = \estimate{\vect{g}}{}{i} -\vect{g}_i$ the estimation errors computed by each $N_j^i \in \Nhop{i}{k}$, i.e.:
\begin{equation}
    \label{Definition: error on input and state estimation of agent i}
        \hspace{-0.2cm}\error{\vect{x}}{}{i} = \left[\error{x}{N_1^i \top}{i}, \dots ,\error{x}{N_{\eta_i}^i \top}{i} \right]^\top\hspace{-0.1cm}, \error{\vect{g}}{}{i}=\left[\error{g}{N_1^i \top}{i}, \dots , \error{g}{N_{\eta_i}^i\top}
    {i} \right]^\top\hspace{-0.1cm}, 
\end{equation}
with $\error{x}{N_j^i}{i} = \estimate{x}{N_j^i}{i} - x_i$ and $\error{g}{N_j^i}{i} = \estimate{g}{N_j^i}{i} - g_i$.

 
To simplify the notation, we assume without loss of generality that $n_i=1$ for all $i\in \mathcal{V}$ in the following sections. Nonetheless, the results can be extended to higher dimensional case by appropriate use of the Kronecker product.

\subsection{Problem formulation}
\label{Problem formulation}
For every agent $i \in \mathcal{V}$ and for all $N_j^i \in \Nhop{i}{k}$, let $\delta^{N_j^i}_i :\mathbb{R}_{\geq 0} \rightarrow \mathbb{R}$ be a prescribed performance function that is used to capture the predefined performance bounds for the estimation errors, as defined in the following:
\begin{definition}
    \label{Definition of Prescribed performance function}
    A function $\rho :\mathbb{R}_{\geq 0} \rightarrow \mathbb{R}$ is a prescribed performance function if it satisfies, for all $t \in \mathbb{R}_{\geq 0}$: (i) $\rho(t) \in \mathcal{C}^1$; (ii)  $0 <\rho(t) \leq \overline{\rho}$ for some $ \overline{\rho} <  \infty $ and (iii) $|\dot{\rho}(t)|\leq \dot{\overline{\rho}}$ for some $\dot{\overline{\rho}}<\infty$. 
\end{definition}

One conventional choice of prescribed performance function is the decreasing exponential function 
\begin{equation}
    \label{eq: choice of prescribed performance function}
    \rho(t)= (\rho(0)-\rho(\infty)) e^{-l t} + \rho(\infty),
\end{equation}
where $\rho(0)$ and $\rho(\infty)$ denote the initial and steady-state values, and $l > 0$ specifies the decay rate.

Then, the goal of this work is formulated as follows.
\begin{problem}
    \label{Problem state estimation error convergence}
     Given the heterogeneous MAS in \eqref{eq: agent's dynamic} communicating over a graph $\mathcal{G}$, and prescribed performance functions $\delta^{N_j^i}_i(t)$, design a decentralized \mbox{$k$-hop} observer such that its estimation errors satisfy the prescribed performance requirements $|\error{x}{N_j^i}{i}(t)| <\delta^{N_j^i}_i(t)$ for all $i \in \mathcal{V}$ and all $N_j^i \in \Nhop{i}{k}$. Furthermore, given a team control objective for the MAS, derive sufficient conditions under which the \mbox{observer-based} decentralized controllers $u_i$ achieve the team objective with arbitrarily small error while using the state estimates.
\end{problem}


\section{Disagreement dynamics}\label{Disagreement dynamics}
For each $i \in \mathcal{V}$ and $N_j^i \in \Nhop{i}{k}$, define the disagreement term $\xi_i^{N_j^i}$ on the estimate of $x_i$ performed by $N_j^i$ as:
\begin{equation}
    \label{eq: disagreement vector on the estimate of agent i}
    \xi_i^{N_j^i} =\sum_{l\in(\Neigh{N_j^i}{} \cap \Nhop{i}{k})} (\statediff{\hat{x}}{N_j^i}{i}{\hat{x}}{l}) + |\Neigh{N_j^i}{} \cap \Neigh{i}{}|  (\hat{x}^{N_j^i}_i - x_i),
\end{equation}
where $\xi_i^{N_j^i}$ represents a local disagreement term capturing how the estimate $\hat{x}^{N_j^i}_i$ differs from: (i) the true state information $x_i$ shared by the agents $l \in \Neigh{N^i_{j}}{} \cap \Neigh{i}{}$ and (ii) the state estimate $\estimate{x}{l}{i}$ shared by those agents $l \in \Neigh{N^i_{j}}{} \cap \Nhop{i}{k}$.

\subsection{Disagreement vector and problem reformulation}
\label{Section: disagreement vector and problem reformulation}
By stacking the disagreement components $\xi_i^{N_j^i}$ for all $N_j^i \in \Nhop{i}{k}$, and using the state estimation error definition in \eqref{Definition: error on input and state estimation of agent i}, the \textbf{disagreement vector} defined as $\vect{\xi}_i := \Bigl[\xi_{i}^{N_1^i}, \dots ,\xi_{i}^{N^i_{\eta_i}}\Bigr]^\top$ can be expressed as:

\begin{equation}
    \label{eq: disagreement vector expression}
    \vect{\xi}_i = (\kc{L}{i} + \kc{H}{i}) \error{\vect{x}}{}{i} = \kc{M}{i} \error{\vect{x}}{}{i},
\end{equation}
where the matrix $\kc{L}{i}$ is the Laplacian matrix of the sub-graph $\mathcal{G}_i = ( \Nhop{i}{k}, \mathcal{E}_i)$ induced by the \mbox{$k$-hop} neighbors of agent $i$, with $\mathcal{E}_i = \{(p,q) \in \mathcal{E}: \{p,q\} \subseteq \Nhop{i}{k}\}$, $\kc{H}{i}:= \text{diag}(|\Neigh{N_1^i}{}\cap \Neigh{i}{}|, \dots, |\Neigh{N_{\eta_i}^i}{}\cap \Neigh{i}{}|) \in \mathbb{R}^{\eta_i \times \eta_i}$ and $M_i^{kc} \in \mathbb{R}^{\eta_i \times \eta_i}$ is defined as $\kc{M}{i} = \kc{L}{i} + \kc{H}{i}$.

\begin{lemma}[\!\!\cite{11187218}]
    \label{Lemma on the positive definiteness of the matrix M}
    If $\mathcal{G}$ is connected, then $\kc{M}{i} \succ 0 $ for all $ i \in \mathcal{V}$ with $ \Nhop{i}{k} \neq \emptyset$.
\end{lemma}
\begin{remark}
    \label{Remark on local disagreement vector computation and positive definiteness of Mik}
    If Assumption~\ref{Assumption on 2-hop communication} does not hold, then due to the unavailability of $x_i$, \eqref{eq: disagreement vector on the estimate of agent i} cannot be computed locally under the proposed $k$-hop definition. Nevertheless, as introduced in Remark \ref{Remark on validity of Assumption 4}, \eqref{eq: disagreement vector on the estimate of agent i} can still be evaluated locally, and the positive definiteness of $\kc{M}{i} = \kc{L}{i} + \kc{H}{i}$ can be preserved by extending the definition of the \mbox{$k$-hop} neighborhood to include the $1$-hop neighbors of each agent. In this case, $\kc{L}{i}$ remains the Laplacian matrix of the subgraph induced by the \mbox{$k$-hop} neighbors of agent~$i$, and $\kc{H}{i} := \text{diag}(h^{N_1^i}_i, \dots, h^{N_{\eta_i}^i}_i) \in \mathbb{R}^{\eta_i \times \eta_i}$, where $h^{N_j^i}_i = 1$ if $N_j^i \in \Neigh{i}{}$.
\end{remark}
\begin{lemma}
    \label{Lemma on boundness of the norm of prescibed performance functions}
    Let $\vect{\rho} (t) = \bigl[\rho_1(t), \dots , \rho_m(t)\bigr]^\top$ be a vector whose components $\rho_i$, $i \in \{1, \dots, m\}$, are prescribed performance functions as per Definition \ref{Definition of Prescribed performance function}. Then, $\norm{\vect{\rho}(t)}: \mathbb{R}_{\geq 0} \rightarrow \mathbb{R}$ is itself a prescribed performance function satisfying conditions (i)-(iii) in Definition \ref{Definition of Prescribed performance function}.
\end{lemma}
\begin{proof}
    (i) Each $\rho_i(t)$ is positive and continuously differentiable ($\mathcal{C}^1$) by definition, hence $\vect{\rho}(t) \in \mathcal{C}^1$. The Euclidean norm is smooth on $\mathbb{R}^{\eta_i} \setminus \{0\}$, and since $\vect{\rho}(t) \neq 0$ for all $t \ge 0$, it follows that $\norm{\vect{\rho}(t)} \in \mathcal{C}^1$. (ii) From Definition~\ref{Definition of Prescribed performance function}, each component satisfies $0 <\rho_i(t) \leq \overline{\rho}_i$ for some $\overline{\rho}_i <  \infty$. Hence, $0 < \norm{\vect{\rho}(t)} \leq \overline{\vect{\rho}}$, with $\overline{\vect{\rho}} = \sqrt{ \sum_{i=1}^{m}(\overline{\rho}_i)^2}$. (iii) Differentiating the Euclidean norm gives
     $\frac{d\norm{\vect{\rho}(t)}}{dt} = \frac{\vect{\rho}^\top(t) \dot{\vect{\rho}}(t)}{\norm{\vect{\rho}(t)}}$. Since $\norm{\vect{\rho}(t)} > 0$, $0 <\rho_i(t) <\infty$, and $|\dot{\rho}_i(t)| <\infty$ for all $i \in \{1, \dots, m\}$, it follows that $\frac{d\norm{\vect{\rho}(t)}}{dt}$ is upper bounded.
\end{proof}

Given the validity of Lemma \ref{Lemma on the positive definiteness of the matrix M}, $\kc{M}{i}$ is always invertible and $\error{\vect{x}}{}{i} = {\kc{M}{i}}^{-1} \vect{\xi}_i$. Thus, from the submultiplicative property, $\norm{\error{\vect{x}}{}{i}} \leq \norm{{\kc{M}{i}}^{-1}} \norm{\vect{\xi}_i}$ holds and $\norm{\error{\vect{x}}{}{i}}$ satisfies $\norm{\error{\vect{x}}{}{i}(t)} \leq \lambda_{\text{min}}^{-1}(\kc{M}{i}) \norm{\vect{\xi}_i(t)}$.

Since $|\error{x}{N_j^i}{i}(t)| \leq \norm{\error{\vect{x}}{}{i}(t)}$ holds from the norm definition, to satisfy Problem \ref{Problem state estimation error convergence} and ensure $|\error{x}{N_j^i}{i}(t)| < \delta_i^{N_j^i}(t)$ for all $N_j^i \in \Nhop{i}{k}$, it suffices to impose $\norm{\vect{\xi}_i(t)} < \lambda_{\text{min}}(\kc{M}{i}) \min_{j \in \{1,\dots, \eta_i\}} \{\delta_i^{N_j^i}(t)\}$ by constraining the evolution of $\xi_i^{N_j^i}$ to satisfy $|\xi_i^{N_j^i}| < \rho_i^{N_j^i}(t)$ for all $N_j^i \in \Nhop{i}{k}$, where each $\rho_i^{N_{j}^i}(t)$ is a prescribed performance function selected such that the norm of $\vect{\rho}_i(t)= \Bigl[\rho_i^{N_1^i}(t),\dots , \rho_i^{N_{\eta_i}^i(t)}\Bigr]^\top$ satisfies $\left\lVert\vect{\rho}_i(t)\right\rVert \leq \lambda_{\text{min}}(\kc{M}{i}) \min_{j \in \{1,\dots, \eta_i\}} \{\delta_i^{N_j^i}(t)\}$. 
As a result, Problem \ref{Problem state estimation error convergence} can be partially reformulated as:
\begin{problem}
    \label{Problem state estimation reformulated}
     Given the heterogeneous MAS in \eqref{eq: agent's dynamic} communicating over a graph $\mathcal{G}$, and prescribed performance functions $\delta^{N_j^i}_i(t)$, design a decentralized \mbox{$k$-hop} observer such that the disagreement dynamics satisfy $|\xi_i^{N_j^i}| < \rho_i^{N_j^i}$ for all $i \in \mathcal{V}$ and all $N_j^i \in \Nhop{i}{k}$, where each $\rho_i^{N_j^i}$ is a prescribed performance function designed so that $\lVert\vect{\rho}_i(t)\rVert \leq \lambda_{\text{min}}(\kc{M}{i}) \min_{j \in \{1,\dots, \eta_i\}} \{\delta_i^{N_j^i}(t)\}$ holds for all $t \in \mathbb{R}_{\geq 0}$.
\end{problem}

\begin{remark}
    Note that, since $\rho_i^{N_j^i}(t)$ are design choices, we can indirectly impose desired behavior to every estimation error $\error{x}{N_j^i}{i}(t)$ by tuning the parameters of $\rho_i^{N_j^i}(t)$.
\end{remark}

For \mbox{multi-dimensional} case, where $n_i \neq 1$, this reasoning can be performed on every component of the agent's state. Hence, desired performance can be imposed on the convergence of every disagreement component of $\xi_{i}^{N_j^i}$, i.e., on every $\xi_{i,l}^{N_j^i}$, with $l \in \{1, \dots, n_i\}$.

\subsection{Prescribed Performance Observer}
\label{Section: Prescribed Performance Observer}
Inspired by the PPC literature \cite{4639441}, we design a $k$-hop Prescribed Performance Observer (PPO) that constrains the disagreement dynamics $\xi_i^{N_j^i}$ to satisfy
\begin{equation}
    \label{eq: definition of the bounds on the disagreement term}
    -\rho^{N_j^i}_i(t) < \xi_{i}^{N_j^i}(t) < \rho^{N_j^i}_i(t)
\end{equation} 
for all $t \in \mathbb{R}_{\geq 0}$, $i \in \mathcal{V}$, and $N_j^i \in \Nhop{i}{k}$, where $\rho^{N_j^i}_i(t)$ is a prescribed performance function, as defined in \eqref{eq: choice of prescribed performance function}, satisfying $\rho^{N_j^i}_i(0) \geq \rho^{N_j^i}_i(\infty) > 0$ and $\rho^{N_j^i}_i(0) > |\xi_{i}^{N_j^i}(0)|$. Given the initial condition $\rho^{N_j^i}_i(0)$, the value $\rho^{N_j^i}_i(\infty) = \lim_{t \rightarrow \infty} \rho^{N_j^i}_i(t)$ represents the maximum allowable magnitude of the disagreement vector at steady state.

Let $e_i^{N_j^i} \in (-1,1)$ denote the normalization of $\xi_{i}^{N_j^i}(t)$ with respect to $\rho^{N_j^i}_i$, i.e., $e_i^{N_j^i} =\rho^{N_j^i}_i(t)^{-1}\xi_{i}^{N_j^i}$, and let $T: (-1,1) \rightarrow \mathbb{R}$ be a strictly increasing transformation satisfying $T(0) = 0$.
For all $i \in \mathcal{V}$ and $N_j^i\in \Nhop{i}{k}$, define the transformed normalized disagreement as:
\begin{equation}
\vspace{-0.05cm}
    \label{eq: transformed normalized disagreement}
    \epsilon^{N_j^i}_i = T(e^{N_j^i}_i) = T(\rho^{N_j^i}_i(t)^{-1}\xi_{i}^{N_j^i}).
\end{equation}
In this work, we select $T(e) = \ln(\frac{1+e}{1-e})$, which has a strictly positive derivative $ J_T(e) = \frac{2}{1-e^2}$.
Then, by defining the transformed normalized disagreement vector as $\vect{\epsilon}_{i}  := \left[\epsilon^{N_1^i}_i, \dots ,\epsilon_{i}^{N^i_{\eta_i}} \right]^\top$, its dynamics result into:
\begin{equation}
    \label{eq: disagreement vector dynamics on agent i state}
     \dot{\vect{\epsilon}}_i = \vect{J}_i  \vect{P}^{-1}_{i}(\dot{\vect{\xi}}_{i} - \dot{\vect{P}}_{i} \vect{e}_i),
\end{equation}
where $\vect{J}_i = \text{diag}\left(J_T(e^{N_1^i}_i), \dots, J_T(e^{N^i_{\eta_i}}_i)\right)$, $\vect{P}_{i} = \text{diag}\left(\rho^{N_1^i}_i, \dots, \rho^{N_{\eta_i}^i}_i\right)$, $\dot{\vect{P}}_{i} = \text{diag}\left(\dot{\rho}^{N_1^i}_i, \dots, \dot{\rho}^{N^i_{\eta_i}}_i\right)$, $\vect{e}_i = \vect{P}^{-1}_i \vect{\xi}_i$ and $\dot{\vect{\xi}}_{i} = \Bigl[\dot{\xi}_{i}^{N_1^i}, \dots ,\dot{\xi}_{i}^{N^i_{\eta_i}}\Bigr]^\top$.

\begin{remark}
    \label{Remark: boundness E impliess funnel satisfaction}
    From \eqref{eq: transformed normalized disagreement}, it follows that if the vector $\vect{\epsilon}_{i}$ is bounded, then $e^{N_j^i}_i$ remains confined within the interval $(-1,1)$ for all $N_j^i \in \Nhop{i}{k}$. Consequently, for every $i\in \mathcal{V}$ and each $N^i_j \in \Nhop{i}{k}$, $\xi_{i}^{N^i_j}$ evolves in compliance with \eqref{eq: definition of the bounds on the disagreement term}.
\end{remark}


\section{$k$-hop Prescribed Performance State Observer}\label{state observer}
In this section a decentralized \mbox{$k$-hop} Prescribed Performance State Observer ($k$-hop PPSO) is introduced to solve Problem \ref{Problem state estimation reformulated}.
In this regard, assume that each agent $N_j^i \in \Nhop{i}{k}$ updates its estimate $\hat{x}_{i}^{N_j^i}$ of the state of agent $i$ as:
\begin{equation}
    \label{eq: component of the state estimation error dyanmics}
    \dot{\hat{x}}_{i}^{N_j^i} = f_i(\estimate{x}{N_j^i}{i}) + \hat{g}^{N_j^i}_i - \rho_i^{N_j^i}(t)^{-1} J_T(e_i^{N_j^i}) \epsilon_{i}^{N_j^i}(t),
\end{equation}
where $\hat{g}^{N_j^i}_i$ is the estimate of $g_i(u_i(t))$ computed by $N_j^i$, and $\rho_i^{N_j^i}(t)$, $J_T(e_i^{N_j^i})$ and $\epsilon_i^{N_j^i}$ are defined as in Section \ref{Section: Prescribed Performance Observer}.
\begin{remark}
    \label{Remark on the decentralization of the observer and on the possibility of removing f and g}
    Note that $\epsilon_{i}^{N_j^i}(t)$, and consequently $\dot{\hat{x}}_{i}^{N_j^i}(t)$, is computed exclusively based on information received from the neighbors of agent $N_j^i$. Hence, provided that each $k$-hop neighbor ${N_j^i}$ of $i$ possesses knowledge of the structure of $f_i$, the proposed observer operates in a fully decentralized manner. Furthermore, as will be demonstrated later, under reasonable assumptions on $g_i(u_i)$, both $f_i(\estimate{x}{N_j^i}{i})$ and $\hat{g}^{N_j^i}_i$ can be omitted from \eqref{eq: component of the state estimation error dyanmics} while still guaranteeing the solution of Problem~\ref{Problem state estimation reformulated}.
\end{remark}

By stacking $\dot{\hat{x}}_{i}^{N_j^i}$ for all $N_j^i \in \Nhop{i}{k}$, the dynamics of $\estimate{\vect{x}}{}{i}$, defined as in \eqref{eq: estimation of state and input of agent i}, becomes:

\begin{equation}
    \label{Eq: State observer dynamics}
    \begin{split}
        \dot{\hat{\vect{x}}}_i &=\vect{f}_i(\estimate{\vect{x}}{}{i}) + \estimate{\vect{g}}{}{i} -  \vect{P}_i^{-1} \vect{J}_i \vect{\epsilon}_i,
    \end{split}
\end{equation}
where $\vect{f}_i(\estimate{\vect{x}}{}{i}) = \Bigl[\state{f}{}{i}(\estimate{x}{N_1^i}{i}), \dots ,\state{f}{}{i}(\estimate{x}{N^i_{\eta_i}}{i})\Bigr]^\top$, $\estimate{\vect{g}}{}{i} = \Bigl[\hat{g}^{N_1^i}_i, \dots, \hat{g}^{N_{\eta_i}^i}_i\Bigr]^\top$ and $\vect{P}_i$, $\vect{J}_i$, and $\vect{\epsilon}_i$ are defined as {in \eqref{eq: disagreement vector dynamics on agent i state}}.

 Assume each agent runs a convergent input observer guaranteeing $\norm{\error{\vect{g}}{}{i}(t)} \leq \delta^{\tilde{g}}_i$ to hold, with $\delta^{\tilde{g}}_i < \infty$. Then, the state estimation errors satisfy the prescribed performance bounds specified in Problem \ref{Problem state estimation error convergence}, as shown in the next result.

    
\begin{theorem}
    \label{Theorem: main theorem on state estimation convergence}
    Consider a heterogeneous MAS \eqref{eq: agent's dynamic} with connected graph $\mathcal{G}$ and decentralized state observers as in \eqref{eq: component of the state estimation error dyanmics}. For all $i \in \mathcal{V}$, assume that the input estimation error $\norm{\error{\vect{g}}{}{i}(t)}$ is upper bounded by some $ \delta^{\tilde{g}}_i \in \mathbb{R}_{\geq 0}$. Then, for all $N_j^i \in \Nhop{i}{k}$ and all $i\in \mathcal{V}$, the state estimation error $\tilde{x}^{N^i_j}_{i}(t)$ satisfies $|\tilde{x}_{i}^{N_j^i}(t)| < \delta^{N_j^i}_i(t)$ provided that $|\xi^{N^i_j}_{i}(0)| < \rho^{N_j^i}_i(0)$ holds for the disagreement terms, and $\rho^{N_j^i}_i(t)$ is designed so that $\lVert\vect{\rho}_i(t)\rVert  \leq \lambda_{\text{min}}(\kc{M}{i}) \min_{j \in \{1,\dots, \eta_i\}} \{\delta_i^{N_j^i}(t)\}$.
\end{theorem}
\begin{proof}    
    Consider agent $i \in \mathcal{V}$. According to Assumption \ref{Assumption on neighbors}, $\mathcal{G}$ is a time invariant graph. Thus, $\kc{M}{i}$ is constant, $\dot{\vect{\xi}}_i = \kc{M}{i} \dot{\tilde{\vect{x}}}_{i}$ from \eqref{eq: disagreement vector expression},  and \eqref{eq: disagreement vector dynamics on agent i state} can be rewritten as:
    \begin{equation}
        \label{eq: transformed error dynamics}
        \dot{\vect{\epsilon}}_i = \vect{J}_i \vect{P}^{-1}_{i}(\kc{M}{i} \dot{\tilde{\vect{x}}}_{i} - \dot{\vect{P}}_{i} \vect{e}_i).
    \end{equation}
    From the agent's dynamic in \eqref{eq: agent's dynamic}, the definitions in \eqref{Definition: error on input and state estimation of agent i} and the observer \eqref{Eq: State observer dynamics}, $\dot{\tilde{\vect{x}}}_{i}$ becomes $\dot{\tilde{\vect{x}}}_{i} = \vect{f}_i(\estimate{\vect{x}}{}{i}) - \vect{f}_i(\vect{x}_{i}) + \error{\vect{g}}{}{i} -  \vect{P}_i^{-1} \vect{J}_i \vect{\epsilon}_i - \vect{w}_i$,
    where $\vect{f}_i(\vect{x}_i) = 1_{\eta_i} \otimes f_i(x_i)$, and $\vect{w}_i = 1_{\eta_i} \otimes w_i(\vect{x},t)$.
    Consider now the candidate Lyapunov function $V = \frac{1}{2} \vect{\epsilon}_i^T \vect{\epsilon}_i$, with time derivative $\dot{V} = \vect{\epsilon}_i^T \dot{\vect{\epsilon}}_i$. By replacing \eqref{eq: transformed error dynamics} and $\dot{\tilde{\vect{x}}}_{i}$, $\dot{V} $ results into:
    \vspace{-0.1cm}
    \begin{equation}
        \label{eq Lyapunvon function derivative}
        \begin{split}
            \dot{V} =& - \vect{\epsilon}_i^T \vect{J}_i \vect{P}^{-1}_{i} \kc{M}{i} \vect{P}_i^{-1} \vect{J}_i \vect{\epsilon}_i + \vect{\epsilon}_i^T \vect{J}_i \vect{P}^{-1}_{i} \Bigl\{\kc{M}{i} [ \error{\vect{g}}{}{i} \\ & + \vect{f}_i(\estimate{\vect{x}}{}{i})- \vect{f}_i(\vect{x}_{i}) - \vect{w}_i] - \dot{\vect{P}}_{i} \vect{e}_i\Bigr\}.
        \end{split}
    \end{equation} 
    Since $\kc{M}{i} \succ 0$ from Lemma \ref{Lemma on the positive definiteness of the matrix M}, $- \vect{\epsilon}_i^T \vect{J}_i \vect{P}^{-1}_{i} \kc{M}{i} \vect{P}_i^{-1} \vect{J}_i \vect{\epsilon}_i \leq - \lambda_{\min}(\kc{M}{i}) \alpha_{J}\alpha_{\rho} \vect{\epsilon}_i^T \vect{\epsilon}_i$ holds with $\alpha_{J} = \min_{N_j^i \in \Nhop{i}{k}} \bigl\{\min_{e^{N_j^i}_i \in (-1,1)} J_T(e^{N_j^i}_i)^2\bigr\}= 4$ and $\alpha_{\rho} = \max_{N_j^i \in \Nhop{i}{k}}\{\max_{t \in \mathbb{R}_{\geq 0}} \rho_i^{N_j^i}(t)^2\}$. From Definition~\ref{Definition of Prescribed performance function}, there exists $\overline{\rho}^{N_j^i}_i < \infty$ such that $\rho_i^{N_j^i}(t) \leq \overline{\rho}^{N_j^i}_i$. Thus, $\alpha_{\rho}$ is bounded as $\alpha_{\rho} \leq \max_{N_j^i \in \Nhop{i}{k}} \{(\overline{\rho}^{N_j^i}_i)^2\}$.    
    By summing and subtracting $\zeta \norm{\vect{P}^{-1}_{i} \vect{J}_i \vect{\epsilon}_{i}}^2$ for some $0 < \zeta < \lambda_{\min}(\kc{M}{i})$, \eqref{eq Lyapunvon function derivative} can be upper bounded as $\dot{V} \leq -(\lambda_{\min}(\kc{M}{i}) - \zeta)\alpha_{J} \alpha_{\rho} \norm{\vect{\epsilon}_i}^2+ \vect{\epsilon}_i^T \vect{J}_i \vect{P}^{-1}_{i}b(t) -\zeta \norm{\vect{P}^{-1}_{i} \vect{J}_i \vect{\epsilon}_{i}}^2$,
    where $b(t) = \kc{M}{i} [ \error{\vect{g}}{}{i} + \vect{f}_i(\estimate{\vect{x}}{}{i})- \vect{f}_i(\vect{x}_{i}) - \vect{w}_i] - \dot{\vect{P}}_{i} \vect{e}_i$.
    By noticing that $\vect{\epsilon}_i^T \vect{J}_i \vect{P}^{-1}_{i}b(t) -\zeta \norm{\vect{P}^{-1}_{i} \vect{J}_i \vect{\epsilon}_{i}}^2$ resemble terms of the quadratic form $\norm{\sqrt{\zeta} \vect{P}^{-1}_{i} \vect{J}_i \vect{\epsilon}_{i} - \frac{1}{2\sqrt{\zeta}}b(t)}^2$, $\dot{V} \leq -(\lambda_{\min}(\kc{M}{i}) - \zeta)\alpha_{J} \alpha_{\rho} \norm{\vect{\epsilon}_i}^2 + \frac{1}{4\zeta}b^\top(t)b(t)$ holds and $\dot{V}$ can be rewritten as
    \vspace{-0.1cm}
    \begin{equation}
        \label{eq: Lyapunov function final upper bound}
        \dot{V} \leq -\kappa V + \vect{b}(t),
    \end{equation}
    with $\kappa = 2(\lambda_{\min}(\kc{M}{i}) - \zeta)\alpha_{J} \alpha_{\rho}$ and $\vect{b}(t) = \frac{1}{4\zeta} \{\lambda_{\max}(\kc{M}{i})[\norm{\vect{f}_i(\estimate{\vect{x}}{}{i})- \vect{f}_i(\vect{x}_{i})} + \norm{\vect{w}_i} + \norm{\tilde{\vect{g}}_i}] + \norm{\dot{\vect{P}}_{i} \vect{e}_i}\}^2$. To proceed, let's check whether $\vect{b}(t)$ admits an upper bound $\Bar{\vect{b}}(t)$.    
    
    Define with $\tilde{\mathcal{X}}_i(t) = \{ \error{\vect{x}}{}{i} \in \mathbb{R}^{\eta_i}| -1_{\eta_i} <\vect{e}_i = \vect{P}^{-1}_i \vect{\xi}_i <1_{\eta_i}\}$ the time varying set containing the state estimation error $\error{\vect{x}}{}{i}$ for which the disagreement terms $\xi_{i}^{N_j^i}(t)$ satisfy the bounds \eqref{eq: definition of the bounds on the disagreement term} for all $N_j^i \in \Nhop{i}{k}$.
    As introduced in Section \ref{Section: disagreement vector and problem reformulation},  $|\error{x}{N_j^i}{i}(t)| \leq \norm{\error{\vect{x}}{}{i}(t)} \leq \lambda_{\text{min}}^{-1}(\kc{M}{i}) \norm{\vect{\xi}_i(t)}$ is valid by construction. Moreover, from $\tilde{\mathcal{X}}_i(t)$ definition, $\norm{\vect{\xi}_i(t)} < \norm{\vect{P}_i 1_{\eta_i}}$ holds in $\tilde{\mathcal{X}}_i(t)$, with $\norm{\vect{P}_i 1_{\eta_i}}$ bounded as a direct result of Definition \ref{Definition of Prescribed performance function}. Thus, since $\norm{\error{\vect{x}}{}{i}(t) } < \lambda_{\text{min}}^{-1}(\kc{M}{i}) \norm{\vect{P}_i 1_{\eta_i}}$ is valid in $\tilde{\mathcal{X}}_i(t)$, $\tilde{\mathcal{X}}_i(t)$ results to be a bounded open set.
    Being $f$ a Lipschitz continuous function, $\norm{\vect{f}_i(\estimate{\vect{x}}{}{i})- \vect{f}_i(\vect{x}_{i})}$ is bounded in $\tilde{\mathcal{X}}_i(t)$. Moreover, since $\dot{\vect{P}}_{i} \vect{e}_{i}$ is a column vector with $ \dot{\rho}^{N_j^i}_i e^{N_j^i}_i$ as entries, and $|\dot{\rho}^{N_j^i}_i(t) e^{N_j^i}_i| < |\dot{\rho}^{N_j^i}_i(t)| < \dot{\overline{\rho}}^{N_j^i}_i$ holds from Definition \ref{Definition of Prescribed performance function} with $\dot{\overline{\rho}}^{N_j^i}_i < \infty$, also $\norm{\dot{\vect{P}}_{i} \vect{e}_i}$ results to be bounded. Then, since $\norm{\tilde{\vect{g}}_i}$ and $\norm{{\vect{w}}_i}$ are bounded by assumption, an upper bound $\Bar{\vect{b}}(t) < \infty$ on $\vect{b}(t)$ is guaranteed to exist for all $ \error{\vect{x}}{}{i} \in \tilde{\mathcal{X}}_i(t)$.
    
    Inspired by \cite[Thm. 22]{10918825}, to prove the invariance of the set $\tilde{\mathcal{X}}_i(t)$, we introduce an auxiliary function $S(\vect{e}_{i}) = 1 - e^{- V(\vect{e}_{i})}$. Note that, from its definition, $S$ satisfies: (i) $S(\vect{e}_{i}) \in (0,1)$ for all $e_{\psi_{l_i}} \in (-1, 1)$, and (ii) $S(\vect{e}_{i}) \rightarrow 1$ as $\vect{e}_{i} \rightarrow \partial \mathcal{D} $, with $\mathcal{D} = \bigtimes_{j=1}^{\eta_i}(-1,1)$. Therefore, studying the boundedness of $\vect{\epsilon}_{i}(\vect{e}_i)$ through the one of $V$ reduces to proving that $S(\vect{e}_{i}) < 1$ holds for all $t$. 
    
    By replacing  \eqref{eq: Lyapunov function final upper bound} and $V(\vect{e}_{i}) = - \ln (1-S(\vect{e}_{i}))$ in $S(\vect{e}_{i})$ derivative, i.e., $\dot{S}(\vect{e}_{i}) = \dot{V}(\vect{e}_{i})(1-S(\vect{e}_{i}))$, $ \dot{S}(\vect{e}_{i}) \leq - \kappa (1-S(\vect{e}_{i})) \Bigl(-\frac{1}{\kappa}\vect{b}(t) - \ln (1-S(\vect{e}_{i}))\Bigr)$ is obtained.
    Since $\kappa$ and $1-S(\vect{e}_{i})$ are positive terms by definition, to verify whether $\dot{S}(\vect{e}_{i}) \leq 0 $ holds, it suffices to study under which conditions $-\frac{1}{\kappa}\vect{b}(t) - \ln (1-S(\vect{e}_{i}))\geq 0$ is valid.
    Note that $-\frac{1}{\kappa}\vect{b}(t) - \ln (1-S(\vect{e}_{i}))\geq 0 $ is satisfied for all $\vect{e}_{i} \in \Omega^c_{\vect{e}}$, where $\Omega_{\vect{e}}=\bigl\{\vect{e}_{i} \in \mathcal{D}| S(\vect{e}_{i}) < 1- e^{-\frac{{\vect{b}}(t)}{\kappa}} \bigr\}$, and that $-\frac{1}{\kappa}\vect{b}(t) - \ln (1-S(\vect{e}_{i})) = 0 $ holds for $\vect{e}_{i} \in  \partial\Omega_{\vect{e}}$. Thus, $\dot{S}(\vect{e}_{i}) \leq 0$  for $\vect{e}_{i} \in \Omega^c_{\vect{e}}$, with $\dot{S}(\vect{e}_{i}) =0$ iff $\vect{e}_{i} \in \partial\Omega_{\vect{e}}$. 
    
    Since the initialization satisfies $|\xi^{N^i_j}_{i}(0)| < \rho^{N_j^i}_i(0)$ for all $N_j^i \in \Nhop{i}{k}$, it follows that $e^{N_j^i}_{i}(0) \in (-1,1)$ for all $N_j^i \in \Nhop{i}{k}$ and therefore that $S(\vect{e}_{i}(0)) < 1$. Moreover, since $e^{-\frac{{\vect{b}}(t)}{\kappa}} \geq e^{-\frac{\Bar{\vect{b}}(t)}{\kappa}} > 0$ by definition, the condition $S(\vect{e}_{i})) < 1 $ is preserved for all $t \in \mathbb{R}_{\geq 0}$, independently of whether $\vect{e}_{i}$ is initialized in $\Omega_{\vect{e}}$ or not. From the inequality $S(\vect{e}_{i})) < 1 $, boundedness of $V(\vect{e}_{i})$, and therefore of the transformed error $\vect{\epsilon}_{i}$, follows. As a result, inequality \eqref{eq: definition of the bounds on the disagreement term} is satisfied.
    If $\rho^{N_j^i}_i(t)$ is designed so that $\lVert\vect{\rho}_i(t)\rVert \leq \lambda_{\text{min}}(\kc{M}{i}) \min_{j \in \{1,\dots, \eta_i\}} \{\delta_i^{N_j^i}(t)\}$ holds, $|\error{x}{N_j^i}{i}(t)| <  \delta^{N_j^i}_i(t)$ is guaranteed by construction for all $N_j^i \in \Nhop{i}{k}$ as explained in Section \ref{Section: disagreement vector and problem reformulation}.
\end{proof}

\begin{remark}
    \label{Remark on input observer avoided}
    Theorem \ref{Theorem: main theorem on state estimation convergence} assumes the existence of a convergent input observer ensuring $\norm{\error{\vect{g}}{}{i}(t)} \leq \delta^{\tilde{g}}_i < \infty$. To relax this assumption, we will propose a $k$-hop Prescribed Performance Input Observer in Section \ref{input observer}.
    Note that, under Assumption \ref{Assumption on g}-(i), the input observer can be omitted, and thus the state observer in \eqref{eq: component of the state estimation error dyanmics} reduces to $\dot{\hat{\vect{x}}}_i =\vect{f}_i(\estimate{\vect{x}}{}{i}) -  \vect{P}_i^{-1} \vect{J}_i \vect{\epsilon}_i$. In this case, satisfaction of the prescribed performance can be proven following the reasoning of Theorem \ref{Theorem: main theorem on state estimation convergence}, while treating $\vect{g}_i$ as a bounded disturbance. Note that, under the assumption of bounded $g_i$, ${g}_i(u_i)$ in \eqref{eq: agent's dynamic} can be extended to treat explicit state dependency, i.e., $g_i(x_i(t), u_i(t))$.
\end{remark}

As mentioned in Remark \ref{Remark on the decentralization of the observer and on the possibility of removing f and g}, under further assumptions on $g_i(u_i)$, \eqref{Eq: State observer dynamics} can be modified to avoid the need of $f_i(\estimate{x}{N_j^i}{i})$.

\begin{theorem}
    \label{Theorem: theorem on state estimation convergence under invariance set}
     Consider a heterogeneous MAS \eqref{eq: agent's dynamic} with connected graph $\mathcal{G}$ and decentralized state observers $\dot{\hat{x}}_{i}^{N_j^i} = \hat{g}^{N_j^i}_i - \rho_i^{N_j^i}(t)^{-1} J_T(e_i^{N_j^i}) \epsilon_{i}^{N_j^i}(t)$ for all $i\in \mathcal{V}$ and $N_j^i \in \Nhop{i}{k}$. For all $i \in \mathcal{V}$, assume that the input estimation error $\norm{\error{\vect{g}}{}{i}(t)}$ is upper bounded by $ \delta^{\tilde{g}}_i \in \mathbb{R}_{\geq 0}$ and that ${g}_{i}(u_i)$ is designed s.t. the agent dynamics as in \eqref{eq: agent's dynamic}
     evolves in a bounded set $\mathbb{X}_i \subset \mathbb{R}^{n_i}$. Then, for all $N_j^i \in \Nhop{i}{k}$ and all $i\in \mathcal{V}$, the state estimation error $\tilde{x}^{N^i_j}_{i}(t)$ satisfies $|\tilde{x}_{i}^{N_j^i}(t)| < \delta^{N_j^i}_i(t)$ provided that $|\xi^{N^i_j}_{i}(0)| < \rho^{N_j^i}_i(0)$ holds for the disagreement terms and $\rho^{N_j^i}_i(t)$ is designed so that $\lVert\vect{\rho}_i(t)\rVert \leq \lambda_{\text{min}}(\kc{M}{i}) \min_{j \in \{1,\dots, \eta_i\}} \{\delta_i^{N_j^i}(t)\}$.
    
\end{theorem}

\begin{proof}
    Consider the candidate Lyapunov function $V = \frac{1}{2} \vect{\epsilon}_i^T \vect{\epsilon}_i$. Following similar steps to those in Theorem \ref{Theorem: main theorem on state estimation convergence}, $\dot{V}$ can be upper bounded as $\dot{V} \leq -\kappa V + \vect{b}(t)$, with $\vect{b}(t) = \frac{1}{4\zeta} \{\lambda_{\max}(\kc{M}{i})[\norm{\vect{f}_i(\vect{x}_{i})} + \norm{\vect{w}_i} + \norm{\tilde{\vect{g}}_i}] + \norm{\dot{\vect{P}}_{i} \vect{e}_i}\}^2$. Then, since $f_i$ is Lipschitz, and ${g}_{i}(u_i)$ ensures the agent dynamics \eqref{eq: agent's dynamic} to evolve in $\mathbb{X}_i \subset \mathbb{R}^{n_i}$, $\norm{\vect{f}_i(\vect{x}_{i})}$ is bounded, and so is $\vect{b}(t)$. Thus, validity of Theorem \ref{Theorem: theorem on state estimation convergence under invariance set} follows by introducing $S(\vect{e}_i)$ as in Theorem \ref{Theorem: main theorem on state estimation convergence}.
\end{proof}

\section{$k$-hop Prescribed Performance Input Observer}\label{input observer}
Even though the results of the previous sections hold under a general $k$-hop input observer, e.g. the one in \cite{11187218}, in this section we propose a decentralized $k$-hop Prescribed Performance Input Observer ($k$-hop PPIO) to estimate each agent's input map $g_i(u_i)$  while guaranteeing $|\error{g}{N_j^i}{i}(t)| < \theta^{N_j^i}_i(t)$ for all $i \in \mathcal{V}$ and all $N_i^j \in \Nhop{i}{k}$, where $\theta^{N_j^i}_i(t)$ is a prescribed performance function as in \eqref{eq: choice of prescribed performance function}.

\subsection{Input disagreement dynamics}
\label{Section on input disagreement dynamics}
Following the design of the $k$-hop PPSO in Sections~\ref{Disagreement dynamics}-\ref{state observer}, let's introduce the input disagreement term on the estimate of $g_i$ performed by $N_j^i$:
\begin{equation}
    \hspace{-0.1cm}
    \label{eq: disagreement vector on the input estimate of agent i}
    \mu_i^{N_j^i} =\sum_{l\in(\Neigh{N_j^i}{} \cap \Nhop{i}{k})} (\statediff{\hat{g}}{N_j^i}{i}{\hat{g}}{l}) + |\Neigh{N_j^i}{} \cap \Neigh{i}{}|  (\hat{g}^{N_j^i}_i - g_i).
\end{equation}
Note that, by stacking $\mu_i^{N_j^i}$ for all $N_i^j \in \Nhop{i}{k}$, a relationship similar to the one in \eqref{eq: disagreement vector expression} holds for $\vect{\mu}_i := \Bigl[\mu_{i}^{N_1^i}, \dots ,\mu_{i}^{N^i_{\eta_i}}\Bigr]^\top$, i.e.:
\vspace{-0.1cm}
\begin{equation}
    \label{eq: disagreement vector on the input function}
    \vect{\mu}_i = \kc{M}{i} \error{\vect{g}}{}{i}.
\end{equation}
Then, for the same reasons specified in Section \ref{Section: disagreement vector and problem reformulation}, to impose $|\error{g}{N_j^i}{i}(t)| < \theta^{N_j^i}_i(t)$ for all $i \in \mathcal{V}$ and $N_i^j \in \Nhop{i}{k}$, it suffices to impose $|\mu_i^{N_j^i}| < \omega_i^{N_j^i}(t)$ for all $N_j^i \in \Nhop{i}{k}$, where each $\omega_i^{N_{j}^i}(t)$ is a prescribed performance function designed so that $\left\lVert\vect{\omega}_i(t)\right\rVert \leq \lambda_{\text{min}}(\kc{M}{i}) \min_{j \in \{1,\dots, \eta_i\}} \{\theta_i^{N_j^i}(t)\}$ holds with $\vect{\omega}_i(t)= \Bigl[\omega_i^{N_1^i}(t),\dots , \omega_i^{N_{\eta_i}^i(t)}\Bigr]^\top$.

Similarly as in Section \ref{Section: Prescribed Performance Observer}, denote with $q_i^{N_j^i} \in (-1,1)$ the normalization of $\mu_{i}^{N_j^i}(t)$ with respect to $\omega_{N_j^i}^i$, i.e., $q_i^{N_j^i} =\omega^{N_j^i}_i(t)^{-1}\mu_{i}^{N_j^i}$, and let $\nu^{N_j^i}_i = T(q^{N_j^i}_i) = T(\omega^{N_j^i}_i(t)^{-1}\mu_{i}^{N_j^i})$ denote the transformed input disagreement.
By defining the transformed input disagreement vector $\vect{\nu}_{i}  := \left[\nu^{N_1^i}_i, \dots ,\nu_{i}^{N^i_{\eta_i}} \right]^\top$, we get:
\vspace{-0.1cm}
\begin{equation}
    \label{eq: disagreement vector dynamics on agent i input}
     \dot{\vect{\nu}}_i = \vect{J}^g_i  \vect{\Omega}^{-1}_{i}(\dot{\vect{\mu}}_{i} - \dot{\vect{\Omega}}_{i} \vect{q}_i),
\end{equation}
where $\vect{J}^g_i = \text{diag}\left(J_T(q^{N_1^i}_i), \dots, J_T(q^{N^i_{\eta_i}}_i)\right)$, $\vect{\Omega}_{i} = \text{diag}\left(\omega^{N_1^i}_i, \dots, \omega^{N_{\eta_i}^i}_i\right)$, $\dot{\vect{\Omega}}_{i} = \text{diag}\left(\dot{\omega}^{N_1^i}_i, \dots, \dot{\omega}^{N^i_{\eta_i}}_i\right)$, $\vect{q}_i = \vect{\Omega}^{-1}_i \vect{\mu}_i$ and $\dot{\vect{\mu}}_{i} = \Bigl[\dot{\mu}_{i}^{N_1^i}, \dots ,\dot{\mu}_{i}^{N^i_{\eta_i}}\Bigr]^\top$.

\begin{remark}
    Similarly to Remark \ref{Remark: boundness E impliess funnel satisfaction}, if $\vect{\nu}_{i}$ is bounded, then $q^{N_j^i}_i$ remains confined within the interval $(-1,1)$ for all $N_j^i \in \Nhop{i}{k}$, and $\mu_{i}^{N^i_j}$ evolves satisfying $|\mu_i^{N_j^i}| < \omega_i^{N_j^i}(t)$.
\end{remark}

\subsection{$k$-hop Prescribe Performance Input Observer design}
Let each agent $N_j^i \in \Nhop{i}{k}$ update its estimate $\hat{g}_{i}^{N_j^i}$ as:
\begin{equation}
    \label{eq: component of the input estimation dyanmics}
    \dot{\hat{g}}_{i}^{N_j^i} = - \omega_i^{N_j^i}(t)^{-1} J_T(q_i^{N_j^i}) \nu_{i}^{N_j^i}(t).
\end{equation}
By staking $\dot{\hat{g}}_{i}^{N_j^i}$ for all $N_j^i \in \Nhop{i}{k}$, the dynamics of $\estimate{\vect{g}}{}{i}$, defined as in \eqref{eq: estimation of state and input of agent i}, becomes:
\vspace{-0.1cm}
\begin{equation}
    \label{Eq: Input observer dynamics}
        \dot{\hat{\vect{g}}}_i = -  \vect{\Omega}_i^{-1} \vect{J}^g_i \vect{\nu}_i,
\end{equation}
where $\vect{\Omega}_i$, $\vect{J}^g_i$, and $\vect{\nu}_i$ are defined as in \eqref{eq: disagreement vector dynamics on agent i input}.

\begin{theorem}
    \label{Theorem: main theorem on input estimation convergence}
    Consider a heterogeneous MAS \eqref{eq: agent's dynamic} with connected graph $\mathcal{G}$ and decentralized input observers as in \eqref{eq: component of the input estimation dyanmics}. Under Assumption \ref{Assumption on g}-(ii), the estimation error $\tilde{g}^{N^i_j}_{i}(t)$ satisfies $|\tilde{g}_{i}^{N_j^i}(t)| < \theta^{N_j^i}_i(t)$  for all $N_j^i \in \Nhop{i}{k}$ and all $i\in \mathcal{V}$, provided that $|\mu^{N^i_j}_{i}(0)| < \omega^{N_j^i}_i(0)$ holds for the disagreement terms, and $\omega^{N_j^i}_i(t)$ is designed so that $\lVert\vect{\omega}_i(t)\rVert \leq \lambda_{\text{min}}(\kc{M}{i}) \min_{j \in \{1,\dots, \eta_i\}} \{\theta_i^{N_j^i}(t)\}$ holds.
\end{theorem}
\begin{proof}
    The proof follows similar arguments to those of Theorem \ref{Theorem: main theorem on state estimation convergence}. Consider the candidate Lyapunov function $V = \frac{1}{2} \vect{\nu}_i^\top \vect{\nu}_i$, whose time derivative is $\dot{V} = \vect{\nu}_i^\top \dot{\vect{\nu}}_i$. From \eqref{eq: disagreement vector dynamics on agent i state}, \eqref{eq: disagreement vector on the input function}, \eqref{Eq: Input observer dynamics} and the definition of $\error{\vect{g}}{}{i}$ in \eqref{Definition: error on input and state estimation of agent i}, $\dot{V}$ becomes $\dot{V} = \vect{\nu}_i^\top  \vect{J}^g_i  \vect{\Omega}^{-1}_{i}\{\kc{M}{i} [\dot{\vect{g}}_i-  \vect{\Omega}_i^{-1} \vect{J}^g_i \vect{\nu}_i] - \dot{\vect{\Omega}}_{i} \vect{q}_i\}$. By adding and subtracting $\zeta \norm{\vect{\Omega}^{-1}_{i} \vect{J}^g_i \vect{\nu}_{i}}^2$ for some $0 < \zeta < \lambda_{\min}(\kc{M}{i})$, and by applying Young’s inequality, it follows that $\dot{V} \leq -\kappa V + \vect{b}(t)$
   holds with $\kappa = 2(\lambda_{\min}(\kc{M}{i}) - \zeta)\alpha_{J} \alpha_{\omega}$, $\vect{b}(t) = \frac{1}{4\zeta} \{\lambda_{\max}(\kc{M}{i})\norm{\dot{\vect{g}}_i} + \norm{\dot{\vect{\Omega}}_{i} \vect{q}_i}\}^2$, $\alpha_{J} = 4$ and  $\alpha_{\omega} = \max_{N_j^i \in \Nhop{i}{k}} \{(\overline{\omega}^{N_j^i}_i)^2\}$, where $\omega^{N_j^i}_i(t) \leq \overline{\omega}^{N_j^i}_i$ holds according to Definition \ref{Definition of Prescribed performance function}.
   Since $\norm{\dot{\vect{g}}_i}$ is bounded by Assumption \ref{Assumption on g}-(ii), and $\norm{\dot{\vect{\Omega}}_{i} \vect{q}_i}$ is bounded for similar reason as $\norm{\dot{\vect{P}}_{i} \vect{e}_i}$ in the proof of Theorem \ref{Theorem: main theorem on state estimation convergence}, $\vect{b}(t)$ is bounded. By introducing $S(\vect{q}_{i}) = 1 - e^{- V(\vect{q}_{i})}$, the proof of Theorem \ref{Theorem: main theorem on input estimation convergence} follows the one of Theorem \ref{Theorem: main theorem on state estimation convergence}. Thus, due to space limitation, the remaining part of the proof is omitted here.
\end{proof}

Theorem \ref{Theorem: main theorem on input estimation convergence} guarantees the estimation error $\error{g}{N_j^i}{i}(t)$ to remain within the prescribed performance bounds. Thus, it provides a way to satisfy the assumption on boundedness of $\norm{\error{\vect{g}}{}{i}(t)}$ required for the validity of Theorem \ref{Theorem: main theorem on state estimation convergence}.

\begin{remark}
    Note that, as in Theorem \ref{Theorem: main theorem on state estimation convergence}, estimation convergence is guaranteed regardless of the upper bound on $\dot{g}_i$.
\end{remark}
\section{$k$-hop PPSO-Based Controller}\label{Closed-loop}
The proposed \mbox{$k$-hop} PPSO allows each agent $i \in \mathcal{V}$ to estimate the state of all agents $N_j^i \in \Nhop{i}{k}$.
This estimation capability enables the synthesis of a closed-loop control law that leverages local state estimates to accomplish the team’s objective.

Consider the vectorized form of the dynamics in \eqref{eq: agent's dynamic}, i.e., $\dot{\vect{x}}(t) = \vect{f}(\vect{x}(t)) + \vect{g}(\vect{u}(t)) + \vect{w}(\vect{x},t)$,
where $\vect{x}(t)= [x_1(t), \dots, x_N(t)]^\top $, $\vect{f}(\vect{x}(t)) = [f_1(x_1), \dots, f_N(x_N)]^\top$, $\vect{g}(\vect{u}(t)) = [g_1(u_1), \dots, g_N(u_N)]^\top$, $\vect{w}(\vect{x},t) = [w_1(\vect{x},t), \dots, w_N(\vect{x},t)]^\top$ and $\vect{u}(t)= [u_1(t), \dots, u_N(t)]^\top $ represents a nonlinear \mbox{state-feedback} control input of the form:
\begin{equation}
\label{input equation closed loop}
    \vect{u} = \bm{\psi}(\vect{x}) = \left[\psi_1(\Bar{\vect{x}}_1,\vect{x}^1),  \dots ,\psi_N(\Bar{\vect{x}}_N,\vect{x}^N) \right]^\top,
\end{equation}
where, for each $i \in \mathcal{V}$, $\vect{x}^i$ is defined as in \eqref{stack vector of real values estimated by agent i}, and  $\Bar{\vect{x}}_i$ contains the state information of agent $i$ and of all $j \in \Neigh{i}{}$.

Since ${\vect{x}}^{i}$ is not locally available, the controller in \eqref{input equation closed loop} is implemented using $\estimate{\vect{x}}{i}{}$ for all $i \in \mathcal{V}$, i.e.:
\begin{equation}
\label{eq: vectorized local input controller}
    \vect{u} = \bm{\psi}(\Bar{\vect{x}},\hat{\vect{x}}) = \left[\psi_1(\Bar{\vect{x}}_1,\hat{\vect{x}}^1),  \dots ,\psi_N(\Bar{\vect{x}}_N,\hat{\vect{x}}^N) \right]^\top.
\end{equation}

Noting that $\Bar{\vect{x}}_i$ consists of components of $\vect{x}$ and that $\estimate{\vect{x}}{i}{} = \vect{x}^i + \error{\vect{x}}{i}{}$ holds by definition, by introducing  $\error{\vect{x}}{}{} = [\error{\vect{x}}{1}{},\dots, \error{\vect{x}}{N}{}]^\top$, $\vect{u}$ can be written as $\vect{u} = \bm{\psi}(\vect{x},\vect{x} + \error{\vect{x}}{}{})$.

\begin{definition}[\!\!\cite{Sontag2008}]
    \label{Definition set-ISS definition fomal}
    A system $\dot{x}=f(x,u,t)$, with $f:\mathbb{R}^n \times \mathbb{R}^m \times \mathbb{R}_{\geq 0} \rightarrow  \mathbb{R}^n$, is set-Input to State Stable (set-ISS) with respect to $u$ if, for each initial condition $x(0)$ and any locally essentially bounded input $u$ satisfying $\sup_{t\geq 0}\norm{u(t)}\leq \infty$, the distance $\norm{x(t)}_{\mathcal{A}} = \inf_{a \in \mathcal{A}}\{\norm{x-a}\}$ of $x(t)$ to $\mathcal{A}$ satisfies $\norm{x(t)}_{\mathcal{A}}\leq \beta(\norm{x(0)}_{\mathcal{A}},t) + \gamma \left(\sup_{0 \leq \tau\leq t}\norm{u(\tau)}\right)$ for all $t \in \mathbb{R}_{\geq 0}$, where $\beta$ and $\gamma$ are a $\mathcal{KL}$ and $\mathcal{K}$ function, respectively.
\end{definition}

\begin{assumption}
    \label{Assumption on regular controller}
    The nominal controller $\vect{u} = \bm{\psi}(\vect{x})$ in \eqref{input equation closed loop} guarantees convergence of the \mbox{multi-agent} system to a set $\mathcal{A}$ representing the team objective, irrespective of the disturbance $\vect{w}(\vect{x},t)$.
\end{assumption}

Assumption \ref{Assumption on regular controller} is not restrictive in practice. Indeed, since $\vect{w}(\vect{x},t)$ is uniformly bounded by Assumption \ref{Assumption on existence of a solution}, \eqref{input equation closed loop} can be designed robustly based on the upper bound of $\vect{w}(\vect{x},t)$.
Therefore, since the nominal controller in \eqref{input equation closed loop} guarantees convergence to the desired set 
$\mathcal{A}$ despite $\vect{w}$, the effect of $\vect{w}$ can be considered as part of the nominal dynamics. Consequently, the behavior of the system under \eqref{eq: vectorized local input controller} can be studied by analyzing the perturbed system $\dot{\vect{x}} = \Phi(\vect{x}, \error{\vect{x}}{}{}, t) = \vect{f}(\vect{x}) + \vect{g}(\bm{\psi}(\vect{x},\vect{x} + \error{\vect{x}}{}{})) + \vect{w}(\vect{x},t)$, where $\error{\vect{x}}{}{}$ is treated as an input disturbance affecting the nominal unforced system.

Let $\delta_{\tilde{\vect{x}}}(t) = \norm{[\delta_{\tilde{x}}^{1}(t)^\top, \dots, \delta_{\tilde{x}}^{N}(t)^\top]^\top}$, with $\delta_{\tilde{x}}^i(t) = [{\delta}^i_{N_1^i}(t), \dots, {\delta}^i_{N_{\eta_i}^i}(t)]^\top$, be the norm of the prescribed performance functions associated with the estimation error $\tilde{\vect{x}}$. Moreover, denote with $\overline{\delta}_{\tilde{\vect{x}}}$ the desired  upper bound on $\norm{\tilde{\vect{x}}}$.

\begin{assumption}
    \label{Assumption on the desired steady state estimation error}
    There exists a finite time $t_x >0$ for which $\delta_{\tilde{\vect{x}}}(t) \leq \overline{\delta}_{\tilde{\vect{x}}}$ holds for all $t \geq t_x$.
\end{assumption}

Since ${\delta}^i_{N_{j}^i}(t)$ are design choices for all $i \in \mathcal{V}$ and all $j \in \Nhop{i}{k}$, they can be designed to satisfy Assumption \ref{Assumption on the desired steady state estimation error}. Thus, Assumption \ref{Assumption on the desired steady state estimation error} is also not restrictive in practice.

\begin{theorem}
    \label{Theorem on closed loop stability}
     Consider a heterogeneous MAS system \eqref{eq: agent's dynamic} with connected graph $\mathcal{G}$ and distributed observers \eqref{eq: component of the state estimation error dyanmics}. Suppose each agent executes the control law in \eqref{eq: vectorized local input controller}. Then, under Assumption \ref{Assumption on the desired steady state estimation error},  the MAS trajectory evolves toward $\mathcal{A}_e = \{\vect{x}: \norm{\vect{x}}_{\mathcal{A}} < \gamma (\overline{\delta}_{\tilde{\vect{x}}})\}$ if $\Phi(\vect{x}, \error{\vect{x}}{}{}, t)$ is \mbox{set-ISS} with respect to $\mathcal{A}$ and the feedback controller in \eqref{input equation closed loop} ensures convergence of the MAS towards $\mathcal{A}$ regardless of $\vect{w}(\vect{x},t)$. 
\end{theorem}

\begin{proof}
     Theorem \ref{Theorem: main theorem on state estimation convergence} guarantees $|\tilde{x}_{i}^{N_j^i}(t)| < \delta_{i}^{N_j^i}(t)$ to hold for all $i \in \mathcal{V}$, $N_j^i \in \Nhop{i}{k}$, and $t \in \mathbb{R}_{\geq 0}$. Thus, under Assumption \ref{Assumption on the desired steady state estimation error}, there exists $t_x$ such that $\norm{\tilde{\vect{x}}(t)} <  \overline{\delta}_{\tilde{\vect{x}}}$ holds for all $t \geq t_x$. Since $\Phi(\vect{x}, \error{\vect{x}}{}{}, t)$ is \mbox{set-ISS} and from $\vect{x}(t_x)$ the system evolves satisfying $\dot{\vect{x}} = \Phi(\vect{x}, \error{\vect{x}}{}{}, t)$ with $\norm{\error{\vect{x}}{}{}} < \overline{\delta}_{\tilde{\vect{x}}}$, $\norm{\vect{x}(t)}_{\mathcal{A}} < \beta(\norm{\vect{x}(t_x)}_{\mathcal{A}},t-t_{x}) + \gamma(\overline{\delta}_{\tilde{\vect{x}}})$ holds $ \forall t \geq t_{x}$. As a result, thanks to the convergence of $\beta(\norm{\vect{x}(t_{x})}_{\mathcal{A}},t-t_{x})$ to zero from $\mathcal{KL}$ function definition, $\vect{x}$ approaches $\mathcal{A}_e =  \{\vect{x}: \norm{\vect{x}}_{\mathcal{A}}< \gamma ( \overline{\delta}_{\tilde{\vect{x}}})\}$ as $t$ goes to infinity. 
\end{proof}

Theorem~\ref{Theorem on closed loop stability} shows that, under Assumptions~\ref{Assumption on regular controller} and~\ref{Assumption on the desired steady state estimation error}, the estimated states can be used in the decentralized controllers $u_i$ to achieve the system objective with a \mbox{worst-case} error governed by $\gamma(\overline{\delta}_{\tilde{\vect{x}}})$. Thus, since $\overline{\delta}_{\tilde{\vect{x}}}$ is a design choice, the desired degree of accuracy can be imposed at design stage.

\section{Simulations}\label{Simulations}
\begin{figure}[t!]
    \vspace{0.1cm}
    \centering
    \includegraphics[width = 0.85\linewidth]{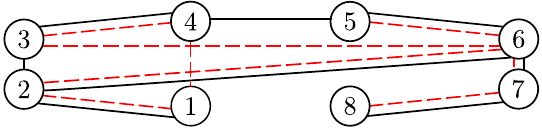}
    \caption{Graphs $\mathcal{G}_C$ and $\mathcal{G}_T$, respectively in solid and dashed lines.}
    \label{fig:Graph Gt used for consensus}
    \vspace{-0.7cm}
\end{figure}
Consider a multi-agent system composed of $N=8$ agents communicating and collaborating, respectively, according to the connected graphs $\mathcal{G}_C=(\mathcal{V},\mathcal{E}_C)$ and $\mathcal{G}_T=~(\mathcal{V},\mathcal{E}_T)$ in Fig. \ref{fig:Graph Gt used for consensus}. Denote with $\mathcal{N}^C_i$ and $\mathcal{N}^T_i$ the $i$-th agent's neighbors in graph $\mathcal{G}_C$ and $\mathcal{G}_T$, respectively. Suppose each agent behaves according to $\dot{x}_i =  f(x_i) + u_i $, where $x_i = [x_{i,1} \ x_{i,2}]^\top$ denotes the state of agent $i$, $f(x_i) = [\tanh(0.5 x_{i,1} + 0.5x_{i,2}) \ \sin(0.5 x_{i,1} - 0.5x_{i,2}) ]^T$ is a Lipschitz continuous function with Lipschitz constant $l_f = 1$, and $u_i =  k_c [\sum_{j\in (\mathcal{N}^C_i  \cap  \mathcal{N}^T_i) } \tanh(x_j-x_i) + \sum_{j\in \mathcal{N}^T_i\setminus\mathcal{N}^{C}_i} \tanh(\hat{x}^i_j-x_i) ]$, with design parameter $k_c$, is a bounded input designed to drive the agents toward consensus by exploiting only the edges of the graph $\mathcal{G}_T$. 
Note that when applied to vectors, $\tanh()$ has to be intended component-wise. The choice of homogeneous dynamics and a simple control objective is done to simplify the verification of the assumptions required to guarantee observer convergence.
In particular, the estimation term $\hat{x}^i_j$, used in $\sum_{j\in \mathcal{N}^T_i\setminus\mathcal{N}^{C}_i}\tanh(\hat{x}^i_j-x_i)$, is introduced to cope with the lack of local information available to agent $i$ about the state of all agents $j \in \mathcal{N}^T_i\setminus\mathcal{N}^{C}_i$. This necessity stems from enforcing consensus using only the edges of $\mathcal{G}_T$, a condition imposed to assess the proposed observers, namely the \mbox{$k$-hop} PPSO and \mbox{$k$-hop} PPIO, which are applied here with $k = 3$ to estimate nonlocal states.

Note that $u_i$ is bounded by definition. Therefore, since Assumption~\mbox{\ref{Assumption on g}-(i)} holds, the $k$-hop PPSO could be simplified by omitting the $k$-hop PPIO. However, because $\dot{u}_i$ is also bounded and Theorem \ref{Theorem: main theorem on input estimation convergence} applies, we retain the full formulation to validate the complete framework. For similar reason, we avoid canceling the nonlinear term $f(x_i)$ with the controller $u_i$.
To analyze stability under ideal input and characterize the set-ISS property of the \mbox{closed-loop} dynamics, define the disagreement projection operator as $\bm{\Pi} := I_{2N} - (1_N 1^\top_N \otimes I_2)/N$ \cite{eb184279-05f5-3acc-a300-750c6f4a17e8}. Accordingly, the consensus disagreement vector is defined as $ \vect{e}_{c} := [e_{\text{c},1}, \dots, e_{\text{c},N}]^\top =  \bm{\Pi} \vect{x}$, the average state as $\bar{x}(t) := \frac{1}{N} \sum_{i=1}^N x_i(t)$, and the average of the nonlinear functions as $\bar{f}(t) := \frac{1}{N} \sum_{i=1}^N f(x_i(t))$. By rewriting the input of every agent as $u_i = u_i^{\text{ideal}} + u_i^{\text{error}}$, where $u_i^{\text{ideal}} :=  k_c \sum_{j\in \mathcal{N}^T_i } \tanh(x_j-x_i)$ and $u_i^{\text{error}} := k_c \sum_{j\in \mathcal{N}^T_i\setminus\mathcal{N}^{C}_i} [ \tanh(\tilde{x}_{j}^{i}+ x_j-x_i) - \tanh(x_j-x_i)]$, and by computing $\dot{\vect{e}}_{c}$, $V = \frac{1}{2} \vect{e}_{c}^\top \vect{e}_{c}$ can be used to prove, following standard Lyapunov-based analysis for consensus, that the MAS achieves consensus under the ideal input, and that the \mbox{closed-loop} system with true input is \mbox{set-ISS} with respect to the consensus manifold \cite{Sontag2008}. As a result, Theorem \ref{Theorem on closed loop stability} holds and the proposed controller is expected to drive the MAS toward consensus. Note that, since $\bar{x}(t)$ evolves as $\dot{\bar{x}}(t) = \frac{1}{N} \sum_{i=1}^N \dot{x}_i(t)$ and $\bar{f}(t) \neq 0$ holds in general, the MAS under ideal inputs $u^{\text{ideal}}_i$ achieves consensus around a \mbox{time-varying} mean. Thus, a similar result is also expected under the true decentralized inputs $u_i$.

For simulation purposes, a sampling time $dt = 10^{-5}$s has been selected.
To guarantee the prescribed performance $|\tilde{x}_{i,l}^{N_j^i}(t)| < \delta(t)$ and $|\tilde{u}_{i,l}^{N_j^i}(t)| < \theta(t)$ to hold, with $\delta(t) = 13.96 e^{-5 t}+0.117 $ and $ \theta(t) = 230 e^{-5 t}+ 1.39 $, $\rho_{N_j^i,l}^i(t)$ and $\omega_{N_j^i,l}^i(t)$ are designed according to Section \ref{Section: disagreement vector and problem reformulation} for all $i\in \mathcal{V}$, $N_j^i \in \Nhop{i}{k}$ and $l \in \{1,2\}$. To satisfy the initialization condition, each component $\omega_{N_j^i,l}^i(0)$ and $\rho_{N_j^i,l}^i(0)$, respectively of $\omega_{N_j^i}^i(0)$ and $\rho_{N_j^i}^i(0)$, has been tuned such that $|\xi^{i}_{N_j^i,l}(0)| < \rho_{N_j^i,l}^i(0)$ and $|\mu^{i}_{N_j^i,l}(0)| < \omega_{N_j^i,l}^i(0)$ hold for all $i \in \mathcal{V}$, $N_j^i \in \Nhop{i}{k}$ and $l \in \{1,2\}$.

Fig. \ref{fig: simulation results}(a) illustrates the \mbox{closed-loop} system behavior when the state estimates provided by the proposed observer are used in the controller. As expected from Theorem \ref{Theorem on closed loop stability}, given the small upper bound on the steady state estimation error imposed by the $k$-hop PPSO, the MAS achieves consensus. However, although the goal is achieved, the agents are not stabilized around a stationary mean. As introduced earlier, this behavior is not caused by the use of estimated states in the control law, but rather by the nonlinearity of $f$.
Given the large number of estimates involved in the network, and due to space limitation,  Fig.~\ref{fig: simulation results}(b) and Fig.~\ref{fig: simulation results}(c) present only the estimation results regarding Agent $4$. For this agent, the performance functions on the disagreement terms have been selected as $\rho_{N_j^4,l}^4(t) = \rho(t) = 2.8 e^{-5 t}+0.02$ and  $\omega_{N_j^4,l}^4(t) = \omega(t) = 39.27 e^{-5 t}+0.033$ for all $N_j^4 \in \Nhop{4}{k}$ and $l \in \{1,2\}$. While Fig.~\ref{fig: simulation results}(b) shows the evolution of the maximum absolute disagreement and estimation error for the estimate of Agent $4$'s state (obtained by agents $N_j^4 \in \Nhop{4}{k}$), Fig.~\ref{fig: simulation results}(c) shows those regarding the input estimates. In accordance with Theorem~\ref{Theorem: main theorem on state estimation convergence} and Theorem~\ref{Theorem: main theorem on input estimation convergence}, Fig. \ref{fig: simulation results} shows that under proper initialization, if $|\xi^{i}_{N_j^i,l}(t)| < \rho_{N_j^i,l}^i(t)$ and $|\mu^{i}_{N_j^i,l}(t)| < \omega_{N_j^i,l}^i(t)$ hold for all $t$, then $|\tilde{x}_{i,l}^{N_j^i}(t)| < \delta(t)$ and $|\tilde{u}_{i,l}^{N_j^i}(t)| < \theta(t)$ are satisfied for all $i\in \mathcal{V}$, $N_j^i \in \Nhop{i}{k}$ and $l \in \{1,2\}$.

\section{Conclusion and Future Work}\label{Conclusion and Future work}
We proposed \mbox{$k$-hop} Prescribed Performance Observers that enable each agent to estimate the state and input of agents up to $k$ hops away while ensuring predefined transient and steady-state performance. The proposed solution is robust to bounded external disturbance and do not require global knowledge of the network. Furthermore, we proved that under set-ISS condition of the feedback control law, the state estimates can be used in the controller to drive the agents toward the team objective.

Future work will address observer design for time-varying and directed graphs.

\begin{figure}[t!]
    \vspace{0.2cm}
    \centering
    \includegraphics[width=1\linewidth]{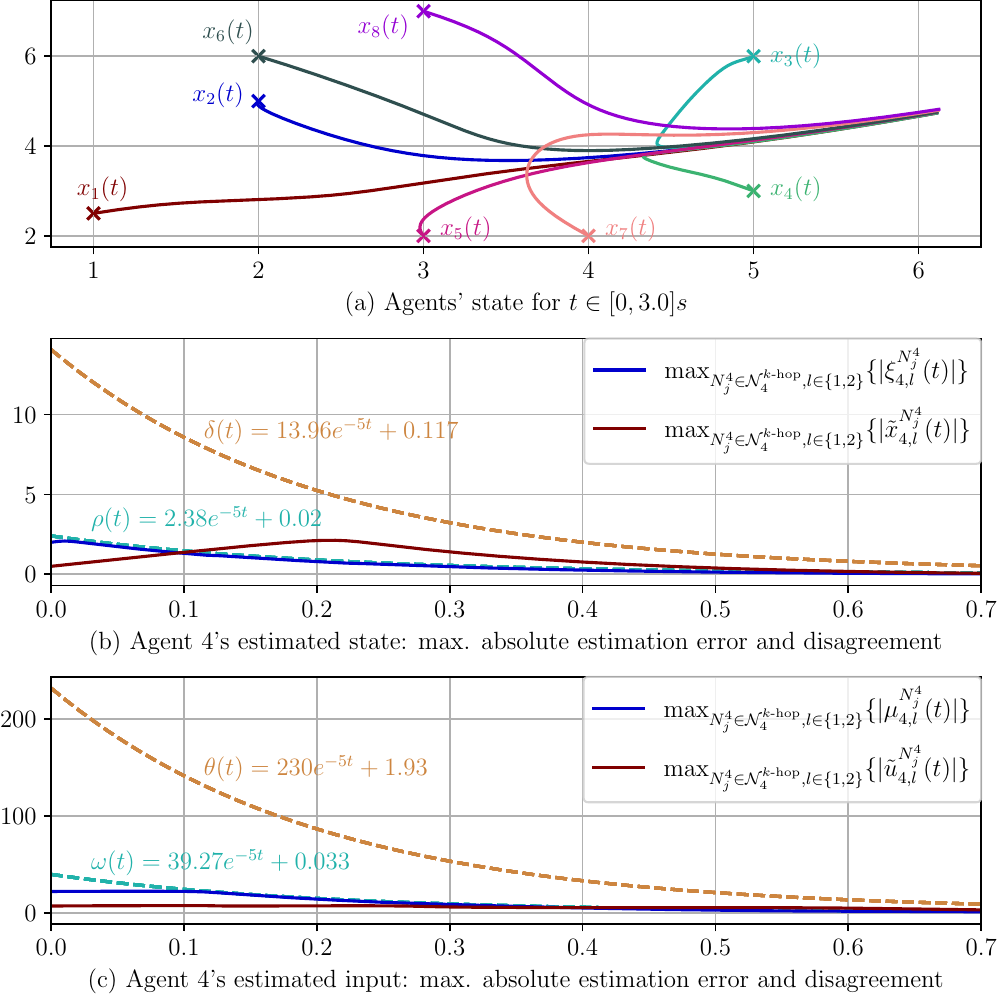}
    \caption{(a) Agents' state evolution; initial states are represented by crosses. (b) and (c) Maximum absolute estimation error and estimation disagreement on Agent 4's state and input. The performance bounds are represented by dashed lines.}
    \label{fig: simulation results}
    \vspace{-0.7cm}
\end{figure}

\bibliographystyle{IEEEtran}
\bibliography{references} 

\end{document}